\documentclass[11pt]{amsart}

\usepackage{graphicx}

\usepackage{newtxtext}
\usepackage{newtxmath}
\usepackage{natbib}
\usepackage{hyperref}
\hypersetup{
    colorlinks = true,
    urlcolor   = blue,
    citecolor  = black,
}
\newtheorem{lemma}{Lemma}

\newtheorem{assumption}{Assumption}
\newtheorem{remark}{Remark}
\newtheorem{theorem}{Theorem}

\newcommand{\RomanNumeralCaps}[1]
\linenumbers

\title[Non-decaying, non-periodic vortex sheets]{The velocity field and Birkhoff-Rott integral for non-decaying, non-periodic vortex sheets}

\author{David M. Ambrose}
\address{Department of Mathematics, Drexel University, Philadelphia, PA 19104 USA}
\email{dma68@drexel.edu}

\begin{document}

\begin{abstract}
The Birkhoff-Rott integral 
expresses the fluid velocity on a vortex sheet.  This integral
converges if certain quantities decay at horizontal infinity, but can also be summed over periodic images
in the horizontally periodic case.  However, non-decaying, non-periodic cases are also of interest,
such as the interaction of periodic wavetrains with non-commensurate periods (i.e. spatially quasiperiodic solutions), 
or non-periodic disturbances to periodic wavetrains.
We  therefore develop a more general single formula for the Birkhoff-Rott integral, which 
unifies and extends the cases of decay and periodicity.
We verify that under some reasonable conditions 
this new version of the Birkhoff-Rott integral is the restriction to the vortex sheet of an incompressible,
irrotational velocity field, with continuous normal component but with a jump in tangential velocity across the 
vortex sheet.  We  give a number of examples of non-decaying, non-periodic sheet positions and sheet strengths
for which our assumptions may be verified.
While we develop this in the case of two-dimensional fluids, the methodology applies equally
well to three-dimensional fluids.
\end{abstract}

\maketitle

\section{Introduction}

A vortex sheet is a surface of discontinuity of velocity in a fluid; in an incompressible fluid, the normal component of velocity across
the sheet must be continuous, but there can be a jump in the tangential component of velocity.  Even with the fluids being irrotational,
this velocity discontinuity, when taking the curl, will lead to vorticity concentrated on the sheet.  
If we approach this from the other direction, then the vortex sheet ansatz
(in the two-dimensional case) is that a fluid has its vorticity, $\omega,$ concentrated along a curve $\mathcal{C},$ parameterized as $(x(\alpha),y(\alpha)).$ 
This vorticity has some amplitude, $\gamma(\alpha),$ known as the vortex sheet strength, which is multiplied by the Dirac mass along the curve.  
That is,
\begin{equation}\nonumber
\omega=\gamma\delta_{\mathcal{C}}.
\end{equation}
This anszatz for the vorticity can then be used with the Biot-Savart law to recover velocity from vorticity; this yields the fluid velocity $(u,v)$ at a
point $\mathbf{x}=(x,y)$ away from the curve $\mathcal{C},$
\begin{equation}\label{velocityOriginal}
(u,v)(\mathbf{x})=\frac{1}{2\pi}\int_{\mathbb{R}}\gamma(\alpha')\frac{(-(y-y(\alpha')),x-x(\alpha'))}{(x-x(\alpha'))^{2}+(y-y(\alpha'))^{2}}\ \mathrm{d}\alpha'.
\end{equation}
We define the Birkhoff-Rott integral, $\mathbf{W}=(W_{1},W_{2}),$ 
by taking the point $\mathbf{x}=(x(\alpha),y(\alpha))$ for some value of $\alpha,$ and interpreting
the integral in the principal value sense:
\begin{equation}\label{BROriginal}
\mathbf{W}(\alpha)=\frac{1}{2\pi}\mathrm{PV}\int_{\mathbb{R}}\gamma(\alpha')
\frac{(-(y(\alpha)-y(\alpha')),x(\alpha)-x(\alpha'))}{((x(\alpha)-x(\alpha'))^{2}+(y(\alpha)-y(\alpha'))^{2}}
\ \mathrm{d}\alpha'.
\end{equation}
The Birkhoff-Rott integral was developed by Birkhoff and Rott in \cite{birkhoff} and \cite{rott}, resepectively, and a
 detailed development of the Birkhoff-Rott integral for two-dimensional fluids may be found, for instance, in the book of \cite{saffmanBook},
or in the book of \cite{majdaBertozzi}.  
Our primary concern in the present work is the question of how to express the velocity field associated to a vortex sheet given by a functions $\gamma$ and 
a curve $\mathcal{C}$ for which \eqref{velocityOriginal} and \eqref{BROriginal} does not obviously make sense. 
These integrals are used either in the spatially decaying case (in which the curve $\mathcal{C}$
is asymptotic at positive and negative infinity to the real line, and in which $\gamma$ decays at positive and negative infinity), and in the spatially 
periodic case; however, there are other settings of interest for vortex sheets, beyond the spatially decaying and spatially periodic cases.  

One clear application
of developing a more general theory would be the ability to study non-periodic perturbations
of periodic wavetrains (and vice versa), or more generally, to study solutions which have both periodic and decaying elements.  As an example
of such a problem, there are known cases of the existence of solitary water waves with capillary ripples, in which the
solution is neither decaying nor periodic (\cite{beale}, \cite{sun}).  The development of a theory which requires neither decay nor periodicity could allow
quite general perturbations of traveling waves to be studied, such as noisy perturbations in which the perturbation is essentially unstructured.

Considering structure other than periodicity, Wilkening and Zhao have carried out detailed numerical studies of irrotational water 
waves (a special case of the vortex sheet) in the spatially quasiperiodic case (\cite{wilkeningZhao1}, \cite{wilkeningZhao2}, 
\cite{wilkeningZhao3}, \cite{wilkeningZhao4}).  These works use the Dirichlet-to-Neumann operator (\cite{craigSulem}) to describe
the fluid velocity rather than the Birkhoff-Rott integral; this is effective in the case that the free fluid surface remains a graph with
respect to the horizontal, or in the two-dimensional case when a conformal mapping may be used.  
See also \cite{dyachenkoSemenova} for further work on spatially quasiperiodic water waves.
We seek instead to be able
to formulate such a problem using a parameterized curve and the Birkhoff-Rott integral, as this generalizes quite naturally to the
three-dimensional setting.  Another non-decaying, non-periodic study of water waves was by Alazard, Burq, and Zuily, where the
initial value problem for irrotational water waves was studied with initial data which is uniformly locally bounded (\cite{alazardEtAl}).
Again, the Dirichlet-to-Neumann operator was used, with a restriction that the free fluid surface is a graph with respect to the
horizontal.  

The main thrust of the present work is to write new expressions for the velocity field $(u,v)$ and the Birkhoff-Rott integral $\mathbf{W};$ 
beyond simply rewriting these formulas, we will give conditions on the
curve $\mathcal{C}$ and on the vortex sheet strength $\gamma$ under which the new integrals converge.  We will then show
that these conditions are satisfied in the usual decaying and spatially periodic cases, and we will be able to see that in these cases 
the new formulas agree with the familiar formulas.

Before proceeding, it is convenient to introduce some complex notation.  If we denote the Birkhoff-Rott integral as $\mathbf{W}=(W_{1},W_{2}),$ 
and if we write $z=x+iy$ and $W=W_{1}+iW_{2},$ and we denote complex
conjugation with $*,$ i.e. $W^{*}=W_{1}-iW_{2},$ then we may write the Birkhoff-Rott integral more succinctly as 
\begin{equation}\label{BRLineComplex}
W^{*}(\alpha)=\frac{1}{2\pi i}\mathrm{PV}\int_{\mathbb{R}}\frac{\gamma(\alpha')}{\xi(\alpha)-\xi(\alpha')}\ \mathrm{d}\alpha'.
\end{equation}
Similarly, we may write the complexified version of \eqref{velocityOriginal} as
\begin{equation}\label{velocityComplex}
u-iv=\frac{1}{2\pi i}\int_{\mathbb{R}}\frac{\gamma(\alpha')}{z-\xi(\alpha')}\ \mathrm{d}\alpha'.
\end{equation}

We define the spatially periodic case to be such that
\begin{equation}\nonumber
x(\alpha+2\pi)=x(\alpha)+2\pi,\qquad y(\alpha+2\pi)=y(\alpha),\qquad \gamma(\alpha+2\pi)=\gamma(\alpha).
\end{equation}
In this case the Birkhoff-Rott integral \eqref{BRLineComplex} may be summed over periodic images:
\begin{multline}\nonumber
W^{*}(\alpha)=\frac{1}{2\pi i}\sum_{k\in\mathbb{Z}}\mathrm{PV}\int_{2\pi k}^{2\pi k + 2\pi}
\frac{\gamma(\alpha')}{\xi(\alpha)-\xi(\alpha')}
\ \mathrm{d}\alpha'
\\
=\frac{1}{2\pi i}\mathrm{PV}\int_{0}^{2\pi}\gamma(\alpha')
\sum_{k\in\mathbb{Z}}\frac{1}{\xi(\alpha)-\xi(\alpha')-2\pi k}\ \mathrm{d}\alpha'.
\end{multline}
So far we have only used complex notation rather than complex analysis.  Now, however, by a 
complex analysis theorem of 
Mittag-Leffler (see \cite{ablowitzFokas}), this sum may be evaluated in closed form.  This gives the formula
\begin{equation}\label{BRPeriodic}
W^{*}(\alpha)=\frac{1}{4\pi i}\mathrm{PV}\int_{0}^{2\pi}\gamma(\alpha')\mathrm{cot}\left(\frac{1}{2}
(\xi(\alpha)-\xi(\alpha'))\right)\ \mathrm{d}\alpha'.
\end{equation}
We denote the circle, or the one-dimensional torus, as $\mathbb{T}=[0,2\pi]$ with periodic boundary conditions,
so \eqref{BRPeriodic} holds for $\alpha\in\mathbb{T}.$
We mention that this summation using the Mittag-Leffler formula must be performed carefully, as the 
infinite sum in question does not converge absolutely.  The summation must be performed symmetrically
in pairs about the $\alpha=\alpha'$ singularity, which is the same as saying the principal value must be taken
at plus and minus infinity.  We will discuss such principal values at infinity more below.
Even though these expressions exist in the periodic case, a goal of the present work is to be able to formulate the velocity integral
and the Birkhoff-Rott integral in a way that does not require special treatment for spatially periodic vortex sheets.
The formulas we will develop will unify the decaying and periodic cases, and apply more generally.  These integrals will therefore
converge in the periodic case without needing to carry out the above summation.

The main idea we utilize in generalizing the classical Birkhoff-Rott integral to more general (non-decaying, non-periodic) vortex sheets
comes from the work of P. Serfati.  In \cite{serfati}, he gave a reformulation of the two-dimensional incompressible Euler equations
allowing for velocity and vorticity in the space $L^{\infty}(\mathbb{R}^{2}),$ and a proof of existence of the solutions of the resulting
equations. The reformulation was in rewriting the Biot-Savart law; specifically, a cutoff function was introduced to separate out the
near-field and far-field contributions of the vorticity to the velocity.  In the far-field, the Biot-Savart integral is non-singular, and 
greater decay can be realized by integrating by parts and thus differentiating the integral kernel.  Ideas inspired by this work of
Serfati have subsequently been used by the author and collaborators in works such as \cite{AKLN} and \cite{ACEK}.
While the current setting is more specialized because of the specific form of vorticity arising in a vortex sheet, we again
introduce a cutoff function to separate the velocity integral into near-field and far-field pieces, and we integrate by parts in the
far-field integral to gain greater decay from the kernel.

After introducing this cutoff function and integrating by parts, we give fairly general conditions under which the new velocity integrals
converge.  We then demonstrate several desirable properties for the new velocity formulation, such as that the velocity is independent of the
particular choice of the cutoff function.  This further includes showing that the velocity field is
divergence-free everywhere and is irrotational away from the vortex sheet, and that the normal component of velocity is continuous across the sheet.
We also show that there is the expected jump in tangential velocity across the sheet, and take the limit of the velocity at vertical infinity.  
In this way we verify that the velocity field is indeed
the velocity field for a vortex sheet.  We also will be able to see that in the decaying case and the horizontally periodic case, our conditions for the 
new velocity integrals to converge do in fact hold, so that we may conclude that the new formulation does properly generalize the 
classical setting.

The plan of the paper is as follows.  In Section \ref{examplesSection} we give four motivating examples of positions of the vortex sheet and
vortex sheet strengths which are not treated by the typical Birkhoff-Rott integral.  In Section \ref{assumptionsSection} we give an assumption
on $\gamma$ and an assumption on $\mathcal{C}$ which will allow us to define our new velocity integral.  We also verify that these assumptions
hold in the decaying and periodic cases.  In Section \ref{newVelocitySection} we introduce our cutoff function and integrate by parts, finding our
new expression for the velocity.  In Section \ref{newBirkhoffRottSection} we then give the corresponding new formula for the Birkhoff-Rott integral.
We show in Section \ref{wellDefinedSection} that under the assumptions, the new velocity and new Birkhoff-Rott integral are well-defined.
We then state an additional assumption on approximability of $\xi$ and $\gamma$ in Section \ref{assumptionCSection}.  We state and prove 
the main theorem, that this velocity is the velocity associated to a vortex sheet, in Section \ref{mainTheoremSection}.  We discuss the jump
in normal and tangential components of the velocity in Section \ref{jumpConditionsSection}, and we take the limit of the velocity field at vertical infinity 
in Section \ref{verticalInfinitySection}.  We revisit the examples of Section \ref{examplesSection}
in Section \ref{examplesAgain}, verifying that the various assumptions hold for these four examples.  We close with some remarks in 
Section \ref{discussionSection}.

\section{Examples and assumptions}

We now will begin the work of deriving a new expression for the velocity \eqref{velocityOriginal} and a new expression for the
Birkhoff-Rott integral. 
These new velocity integrals will agree with the previous forms 
in both the decaying case and the periodic case, but will be single formulas (one formula for the velocity in the plane
and one formula for the Birkhoff-Rott integral) which apply not only
in these two cases, but also to non-decaying, non-periodic cases under some reasonable conditions.  
We wish to first give some illustrative examples of such non-decaying, non-periodic curves and vortex sheet strengths to which
our new formulas for the velocity will apply.

\subsection{Four illustrative examples}\label{examplesSection}

It will be helpful to have in mind a few examples of non-decaying, non-periodic vortex sheets for which the new framework will be
useful.  We give here four examples of curves $\mathcal{C}$ and vortex sheet strengths $\gamma.$  We will in subsequent
sections be making assumptions about $\mathcal{C}$ and $\gamma,$ and in Section \ref{examplesAgain} below we will verify that 
these four examples satisfy all of the assumptions to be made.

\begin{description}

\item[(a) A mixed case.]\ As a simple example we could consider a mixed decaying and periodic case.  For instance, the curve
$\mathcal{C}$ could be taken to be periodic, say $\xi(\alpha)=\alpha+i\sin(\alpha),$ while the vortex sheet strength could be
decaying, say $\gamma(\alpha)=\left(\frac{1}{1+\alpha^{2}}\right)^{2/5}.$

\

\item[(b) A quasiperiodic case.]\ As a simple spatially quasiperiodic example, for some amplitude parameter $\mu>0,$ we could
take the curve $\mathcal{C}$ to be specified by its parameterization 
\begin{equation}\nonumber
\xi(\alpha)=(\alpha+\mu\sin(\alpha)+\mu\sin(\pi\alpha))+i\mu\sin\left(\sqrt{2}\alpha\right).
\end{equation}
Here, the curve is non-self-intersecting if the parameter $\mu$ is sufficiently small.
The vortex sheet strength could be given by, for instance, $\gamma(\alpha)=\cos\left(\left(1+4\pi-\sqrt{2}\right)\alpha\right).$

\

\item[(c) A uniformly local case.]\ As a non-periodic, non-decaying example which also does not have the above quasiperiodic structure, we could
take the curve $\mathcal{C}$ to have parameterization
\begin{equation}\nonumber
\xi(\alpha)=\alpha+i\sin\left(\left(\frac{\alpha^{4}}{1+\alpha^{2}}\right)^{1/4}\right).
\end{equation}
We could take any appropriate $\gamma,$ for example $\gamma(\alpha)=\frac{\alpha^{2}}{1+\alpha^{2}}.$
\

\item[(d) A bore.]\ We let the curve $\mathcal{C}$ be described by its parameterization 
$\xi(\alpha)=\alpha+i\tan^{-1}(\alpha),$ so that the curve $\mathcal{C}$ has different 
horizontal limits on the left and on the right.  We could take the vortex sheet strength to be given by $\gamma(\alpha)=1+\sin(\alpha),$ 
for instance.

\end{description}

\subsection{Assumptions on $\gamma$ and $\mathcal{C}$}\label{assumptionsSection}

We will now state an assumption on $\gamma$ and an assumption on the curve $\mathcal{C}$ which
will allow the generalized formulas to be developed.

\begin{assumption}\label{assumptionA}
We say that $\gamma$ satisfies Assumption \ref{assumptionA} if there exists a constant $c$ and an
antiderivative $\Gamma$ of $\gamma-c$ such that for some $\varepsilon>0,$ there exist constants $C>0$ and 
$0\leq\beta<1$ such that for all $\alpha$ with $|\alpha|>\varepsilon,$
\begin{equation}\label{assumptionABound}
\frac{|\Gamma(\alpha)|}{|\alpha|^{\beta}}\leq C.
\end{equation}
\end{assumption}

Assumption \ref{assumptionA} states that the antiderivative of $\gamma$ grows sublinearly, although a constant may be subtracted first.
The periodic case demonstrates the necessity of subtracting a constant; constant functions have linearly growing antiderivative, so if we
want sublinear growth in the antiderivative we must eliminate this constant first.  The need for a sublinearly growing antiderivative will be 
seen below in \eqref{newVelocityForm}, our new expression for the velocity induced by a vortex sheet.

We next see that the constant, $c,$ to be subtracted from $\gamma$ is uniquely determined.

\begin{lemma}\
Let $\gamma$ satisfy Assumption \ref{assumptionA}.  The constant, $c,$ associated to
$\gamma$ and $\Gamma$ is uniquely defined.  We denote the constant as $c=c[\gamma].$
\end{lemma}
The proof of this lemma is straightforward, since attempting to use a different choice of constant would cause
linear growth in the antiderivative.

We next see that Assumption \ref{assumptionA} is satisfied in both the decaying and periodic cases.

\begin{lemma} \label{assumptionASatisfiedLemma}
If $\gamma\in L^{2}(\mathbb{R})$ or $\gamma\in L^{2}(\mathbb{T}),$ then Assumption
\ref{assumptionA} is satisfied.
\end{lemma}

We prove this lemma in the appendix.

We will be manipulating the classical Birkhoff-Rott integral to give our new version of it, and in so doing, as suggested in the
statement of Assumption \ref{assumptionA}, we will be adding and subtracting $c[\gamma];$ this means one term we must 
deal with is
\begin{equation}\nonumber
\frac{1}{2\pi i}\mathrm{PV}\int_{-\infty}^{\infty}\frac{c[\gamma]}{\xi(\alpha)-\xi(\alpha')}\ \mathrm{d}\alpha'.
\end{equation}
We now state an assumption; when we make this assumption, we will be assuming that this integral makes sense. 
We verify in Lemma \ref{assumptionBLemma} that this assumption
holds in the usual decaying or periodic cases.

\begin{assumption}\label{assumptionB}
We say that the curve $\mathcal{C}$ satisfies Assumption \ref{assumptionB} if the following hold:
\begin{enumerate}
\item for all $\alpha\in\mathbb{R},$
the integral
\begin{equation}\label{assumptionBDisplay}
\mathrm{PV}\int_{-\infty}^{\infty}\frac{1}{\xi(\alpha)-\xi(\alpha')}\ \mathrm{d}\alpha'
\end{equation}
converges in the principal value sense, with the principal value being taken both at $\alpha'=\alpha$ and at $\alpha'=\pm\infty,$
\item for all $z\in\mathbb{C}$ such that $\mathrm{dist}(z,\mathcal{C})>0,$ the principal value integral
\begin{equation}
\mathrm{PV}_{\infty}\int_{-\infty}^{\infty}\frac{1}{z-\xi(\alpha')}\ \mathrm{d}\alpha'
\end{equation}
converges.
\end{enumerate}
\end{assumption}

\begin{remark} We clarify what we mean by the principal value at infinity.  If we are only taking the principal value at infinity, when we write
\begin{equation}\nonumber
\mathrm{PV}_{\infty}\int_{-\infty}^{\infty}g(s)\ \mathrm{d}s,
\end{equation}
this indicates that there exists $a\in\mathbb{R}$ and $M_{n}\nearrow\infty$ such that
\begin{equation}\label{PVInfinity}
\lim_{n\rightarrow\infty}\int_{a-M_{n}}^{a+M_{n}}g(s)\ \mathrm{d}s
\end{equation}
exists, and the principal value integral is equal to the value of this limit; furthermore, if there are multiple choices of $a$ and the sequence $M_{n},$ 
then the principal value at infinity is only well-defined if all such choices will give the same limit in \eqref{PVInfinity}. 
We will sometimes identify the particular $a\in\mathbb{R}$ by
writing, for example,
\begin{equation}\nonumber
\mathrm{PV}_{\infty}\left(\int_{-\infty}^{a}g(s)\ \mathrm{d}s+\int_{a}^{\infty}g(s)\ \mathrm{d}s\right).
\end{equation}
We also could exclude a portion of the domain and write
\begin{equation}\nonumber
\mathrm{PV}\left(\int_{-\infty}^{a}g(s)\ \mathrm{d}s+\int_{b}^{\infty}g(s)\ \mathrm{d}s\right).
\end{equation}
This indicates that we take the similar limit,
\begin{equation}\nonumber
\lim_{n\rightarrow\infty}\int_{a-M_{n}}^{a}g(s)\ \mathrm{d}s + \int_{b}^{b+M_{n}}g(s)\ \mathrm{d}s.
\end{equation}
If we take the principal value at both a finite value and also at infinity, we simply write $\mathrm{PV}$ (as in \eqref{assumptionBDisplay}).
\end{remark}

It is frequently the case when considering the curve $\mathcal{C}$ that we need it to be non-self-intersecting.  We typically accomplish
this by imposing the chord-arc condition, as in \cite{wu2D}.  This condition is
\begin{equation}\label{chordArc}
\exists K>0 \ \mathrm{such}\ \mathrm{that}\ \forall \alpha,\alpha',
\quad \left|\frac{\xi(\alpha)-\xi(\alpha')}{\alpha-\alpha'}\right| \geq K.
\end{equation}

\begin{lemma}\label{assumptionBLemma} 
If $\mathcal{C}$ satisfies the chord-arc condition \eqref{chordArc} and $\xi(\alpha)-\alpha$
is twice continuously differentiable and bounded, and is either in $L^{2}(\mathbb{R})$ or $L^{2}(\mathbb{T}),$ 
and if there exists $M>0$ such that for all $\alpha\in\mathbb{R},$
$\frac{1}{M}\leq|\xi_{\alpha}(\alpha)|\leq M,$ 
then Assumption \ref{assumptionB} is satisfied.  
\end{lemma}

Again, we postpone the proof of this lemma until the appendix.

\section{The new expression for the velocity}

In this section we will form our new expression for the velocity field associated to a vortex sheet, and we will also give our new
expression for the Birkhoff-Rott integral.  We will then show that, under some assumptions, this velocity is divergence-free, is
curl-free in the interior of the two fluid regions, and satisfies the desired jump conditions at the fluid interface.  These formulas
will use a cutoff function which we will introduce, and we will also show that under the same assumptions, the velocity is independent
of the choice of cutoff function.

\subsection{Rewriting the velocity}\label{newVelocitySection}
We introduce a cutoff function, $a:\mathbb{R}\rightarrow\mathbb{R}.$  We will take $a$ to be of class $C^{\infty},$ with 
$a(s)=1$ for $|s|\leq1$ and $a(s)=0$ for $|s|\geq2.$  On the interval $s\in[-2,-1]$ we take $a$ to be monotone increasing, 
and on the interval $s\in[1,2]$ we take $a$ to be monotone decreasing.

We assume Assumption \ref{assumptionA} and Assumption \ref{assumptionB}.
With $c[\gamma]$ as in Assumption \ref{assumptionA}, we denote $\tilde{\gamma}=\gamma-c[\gamma].$
 Then, we first decompose the velocity \eqref{velocityComplex} for $z\notin\mathcal{C}$ as 
\begin{equation}\nonumber
(u-iv)(z)=\frac{1}{2\pi i}\mathrm{PV}_{\infty}\int_{-\infty}^{\infty}\frac{c[\gamma]}{z-\xi(\alpha')}\ \mathrm{d}\alpha'
+\frac{1}{2\pi}\mathrm{PV}_{\infty}\int_{-\infty}^{\infty}\frac{\tilde{\gamma}(\alpha')}{z-\xi(\alpha')}\ \mathrm{d}\alpha'.
\end{equation}
We have written these as principal value integrals at infinity in case this is necessary for convergence (as in the spatially periodic case,
for example).
We next introduce the cutoff function into the second integral on the right-hand side:
\begin{multline}\nonumber
(u-iv)(z)=\frac{1}{2\pi i}\mathrm{PV}_{\infty}\int_{-\infty}^{\infty}\frac{c[\gamma]}{z-\xi(\alpha')}\ \mathrm{d}\alpha'
+\frac{1}{2\pi}\int_{-\infty}^{\infty}\frac{a(|z-\xi(\alpha')|^{2})\tilde{\gamma}(\alpha')}{z-\xi(\alpha')}\ \mathrm{d}\alpha'
\\
+\frac{1}{2\pi}\mathrm{PV}_{\infty}\int_{-\infty}^{\infty}\frac{(1-a(|z-\xi(\alpha')|^{2}))\tilde{\gamma}(\alpha')}{z-\xi(\alpha')}\ \mathrm{d}\alpha'.
\end{multline}
Notice that the second integral on the right-hand side is no longer a principal value, as the cutoff function localizes the integral.
As we discussed in the introduction, the point of introducing the cutoff function is to be able to realize greater decay from the kernel
when integrating by parts in the far-field piece.  We therefore integrate by parts in the third integral on the right-hand side, using the
identity $\Gamma_{\alpha}=\tilde{\gamma},$ finding
\begin{multline}\label{newVelocityForm}
(u-iv)(z)=\frac{1}{2\pi i} \mathrm{PV}_{\infty}\int_{-\infty}^{\infty}\frac{c[\gamma]}{z-\xi(\alpha')}\ \mathrm{d}\alpha'
+\frac{1}{2\pi i}\int_{-\infty}^{\infty}\frac{a(|z-\xi(\alpha')|^{2})\tilde{\gamma}(\alpha')}{z-\xi(\alpha')}\ \mathrm{d}\alpha'
\\
-\frac{1}{2\pi i}\int_{-\infty}^{\infty}\frac{(a'(|z-\xi(\alpha')|^{2}))\left(\mathrm{Re}(\xi_{\alpha}(\alpha')(\overline{z-\xi(\alpha')}))\right)
\Gamma(\alpha')}{z-\xi(\alpha')}\ \mathrm{d}\alpha'.
\\
-\frac{1}{2\pi i}\int_{-\infty}^{\infty}\frac{(1-a(|z-\xi(\alpha')|^{2})\xi_{\alpha}(\alpha')\Gamma(\alpha')}{(z-\xi(\alpha'))^{2}}\ \mathrm{d}\alpha'.
\end{multline}
Notice that the final two integrals on the right-hand side are again no longer principal value integrals; the first of these has now been localized,
while the final integral converges without the principal value.  
We introduce further notation following this decomposition,
\begin{equation}\nonumber
u-iv=I_{0}+I_{1}+I_{2}+I_{3}.
\end{equation}

\subsection{The new expression for the Birkhoff-Rott integral}\label{newBirkhoffRottSection}

By replacing $z$ in the velocity integrals with a point $\xi(\alpha)$ on $\mathcal{C},$
we define our extension of the Birkhoff-Rott integral to be
\begin{equation}\label{Wdecomposition}
W=W_{0}+W_{1}+W_{2}+W_{3},
\end{equation}
where each of these integrals are defined as
\begin{equation}\nonumber
W_{0}^{*}=\frac{1}{2\pi i}\mathrm{PV}\int_{-\infty}^{\infty}\frac{c[\gamma]}{\xi(\alpha)-\xi(\alpha')}\ \mathrm{d}\alpha',
\end{equation}
\begin{equation}\nonumber
W_{1}^{*}=\frac{1}{2\pi i}\int_{-\infty}^{\infty}\frac{a(|\xi(\alpha)-\xi(\alpha')|^{2})\tilde{\gamma}(\alpha')}{\xi(\alpha)-\xi(\alpha')}\ \mathrm{d}\alpha',
\end{equation}
\begin{equation}\nonumber
W_{2}^{*}=-\frac{1}{2\pi i}\int_{-\infty}^{\infty}\frac{(a'(|\xi(\alpha)-\xi(\alpha')|^{2}))\left(\mathrm{Re}(\xi_{\alpha}(\alpha')(\overline{\xi(\alpha)-\xi(\alpha')}))\right)
\Gamma(\alpha')}{\xi(\alpha)-\xi(\alpha')}\ \mathrm{d}\alpha',
\end{equation}
\begin{equation}\nonumber
W_{3}^{*}=-\frac{1}{2\pi i}\int_{-\infty}^{\infty}\frac{(1-a(|\xi(\alpha)-\xi(\alpha')|^{2})\xi_{\alpha}(\alpha')\Gamma(\alpha')}{(\xi(\alpha)-\xi(\alpha'))^{2}}\ \mathrm{d}\alpha'.
\end{equation}

\subsection{The new velocity is well-defined}\label{wellDefinedSection}

Under some assumptions, we can ensures that the new expressions for the velocity integral and the Birkhoff-Rott integral are well-defined.  We now detail
the assumptions, and then we discuss how this gives convergence of the relevant integrals.

We take $\gamma$ to be locally $L^{2}$ integrable, and we take the curve $\mathcal{C}$ to have twice continuously differentiable parameterization $\xi.$
We assume Assumption \ref{assumptionA} and we either assume Assumption \ref{assumptionB} or that $c[\gamma]=0.$  We assume that
$\xi$ satisfies the chord-arc condition \eqref{chordArc}.  We take the parameterization $\xi$ to be such that 
there exists $M>0$ such that for all $\alpha\in\mathbb{R},$ $\frac{1}{M}\leq |\xi_{\alpha}(\alpha)|\leq M.$
Finally, we assume that the function $\xi(\alpha)-\alpha$ may be decomposed as $f_{1}+f_{2},$ with $f_{1}\in L^{2}(\mathbb{R})$ and $f_{2}\in L^{\infty}(\mathbb{R}).$

We mention the reason that we take $\xi(\alpha)-\alpha$ to be able to be written as the sum of a function in $L^{2}$ plus a bounded function.  
It is because either of these assumption on their own is helpful; in the decaying case we have $\xi(\alpha)-\alpha$ in $L^{2},$ and the velocity is
well-defined in this case.  In the periodic case we typically have enough regularity that $\xi(\alpha)-\alpha$ is bounded, and the velocity is well-defined
in this case.  We can work with either of these, and to be more general we make the assumption that there is a part, $f_{1},$ which is in $L^{2},$ and
a part, $f_{2},$ which is bounded.

We now discuss how these assumptions give convergence of the integrals.
By Assumption \ref{assumptionA}, the constant $c[\gamma]$ exists.  Either $c[\gamma]=0,$ in which case $I_{0}$ and $W_{0}$ exist, or
we may rely upon Assumption \ref{assumptionB} which tells us that $I_{0}$ and $W_{0}$ exist.

We next consider $I_{1}$ and $I_{2}.$  By the chord-arc condition \eqref{chordArc} we see that $|z(\alpha)|\rightarrow\infty$ as $\alpha\rightarrow\pm\infty.$
With the presence of the cutoff function in $I_{1}$ and $I_{2},$ we infer that for fixed $z,$ these could be written as integrals over only a finite interval of values
of $\alpha.$  Because the distance from $z$ to $\mathcal{C}$ is non-zero, the denominators in these integrals are not singular.  The assumption that
$\gamma$ is locally $L^{2}$ integrable is now enough to ensure that $I_{1}$ and $I_{2}$ converge.

For $I_{3},$ we  add and subtract to see
\begin{equation}\nonumber
\frac{1}{z-\xi(\alpha')}=\frac{1}{z-\alpha'}+\frac{\xi(\alpha')-\alpha'}{(z-\alpha')(z-\xi(\alpha'))}.
\end{equation}
If the final numerator is bounded, then since $z$ is a positive distance from $\mathcal{C},$ this behaves like $\frac{1}{\alpha'}$ as 
$|\alpha'|\rightarrow\infty.$  Together with the sublinear growth of $\Gamma$ guaranteed by Assumption \ref{assumptionA}, this is enough to ensure
that this integral converges.  If instead the final numerator is in $L^{2}(\mathbb{R}),$ then we may instead use that the final fraction is a product of two
$L^{2}$ functions and is thus integrable.  (Note that we can without loss of generality assume $z\neq\alpha'$ for any $\alpha'$ in the relevant domain of integration;
this is because if we did have $z=\alpha'$ then we could decompose the integral into a piece with a finite domain of integration and a new improper integral
excluding such value of $\alpha'.$  The integral with finite domain of integration then becomes similar to $I_{2}.$)

For $W_{1},$ we consider the decomposition
\begin{equation}\nonumber
\frac{1}{\xi(\alpha)-\xi(\alpha')}=\frac{1}{\xi_{\alpha}(\alpha')(\alpha-\alpha')}+\left[\frac{1}{\xi(\alpha)-\xi(\alpha')}-\frac{1}{\xi_{\alpha}(\alpha')(\alpha-\alpha')}\right].
\end{equation}
The first term on the right-hand side gives rise to a localized Hilbert transform which converges since we have assumed $\gamma$ is locally $L^{2}$ integrable.
The second term on the right-hand side (in brackets) is expressible in terms of divided differences, as in the proof of Lemma \ref{assumptionBLemma}
(see Appendix \ref{proofsSection}).  
This term in non-singular
(i.e. the singularity has been subtracted), and the integral is finite.

For $W_{2}$ and $W_{3},$ we use the chord-arc condition \eqref{chordArc}, since
\begin{equation}\nonumber
\left|\frac{1}{\xi(\alpha)-\xi(\alpha')}\right|=\left|\frac{\alpha-\alpha'}{\xi(\alpha)-\xi(\alpha')}\right| \left|\frac{1}{\alpha-\alpha'}\right| \leq \frac{1}{K}
\left|\frac{1}{\alpha-\alpha'}\right|.
\end{equation}
These integrals are non-singular because of the presence of the cutoff function.  Specifically, these integrals are only supported at values of $\alpha'$ such that
$|\xi(\alpha)-\xi(\alpha')|\geq 1,$ and thus for any given $\alpha$ there is a neighborhood of $\alpha$ for which the integrand is zero.  
Now, because Assumption \ref{assumptionA} implies that $\Gamma$ grows sublinearly at positive and negative infinity, these integrals can be seen to converge.

\subsection{Approximation by decaying vortex sheets and uniform convergence on compact sets}\label{assumptionCSection}

We next wish to conclude that the velocity field given by \eqref{newVelocityForm} is indeed the velocity field for a vortex sheet.  However, it is difficult to see directly
that the divergence and curl of the velocity in this form are zero in the interior of the fluid regions.  Instead, we will use an approximation scheme, then
compute the divergence and curl for the approximate solutions, study convergence of the approximations, and then take the limit of the divergence and 
curl in the sense of distributions.
We assume that the curve $\xi$ and the modified vortex sheet strength $\tilde{\gamma}$ can be well-approximated by decaying curves $\xi^{n}$ and
and decaying functions $\tilde{\gamma}^{n}$ in a certain sense.

Before discussing this approximation scheme, we make some remarks on the parameterization $\xi$ of the curve $\mathcal{C}.$
First, we have assumed that the parameterization is such that $|\xi_{\alpha}|$ is bounded above and below (away from zero); it is 
always possible to find such a parameterization, for instance by parameterizing by arclength.
Next, the chord-arc condition \eqref{chordArc} implies that $|\xi(\alpha)|\rightarrow\infty$ as $|\alpha|\rightarrow\infty,$ and $\xi$ and $\alpha$ go to infinity at 
comparable rates.    As a consequence of these facts, for any compact
set, $K_{0},$ for large enough $|\alpha|,$ we have that $\xi(\alpha)\notin K_{0}.$  These are some features which we seek to maintain in making an approximation 
by decaying vortex sheets.  Furthermore, we wish to have the property that for some interval of values $[\alpha_{-n},\alpha_{n}],$ the approximation $\xi^{n}$
and the original parameterization $\xi$ will agree, with $\alpha_{-n}\rightarrow-\infty$ and $\alpha_{n}\rightarrow\infty$ as $n\rightarrow\infty.$

Specifying such an approximation is complicated in general, as the interface may have complicated overturning features.
The approximation of $\mathcal{C}$ by decaying vortex sheets is most easily achievable in the graph case, when $\xi(\alpha)=\alpha+i\eta(\alpha).$
For any $n\in\mathbb{N},$ we introduce a cutoff function $\phi_{n}.$  This satisfies $\phi_{n}(\alpha)=1$ for $\alpha\in[-n,n],$ and $\phi_{n}(\alpha)=0$
for $|\alpha|>n+1.$  Otherwise, $\phi_{n}$ is smooth and monotone.  
Then we let $\xi^{n}(\alpha)=\alpha+i\eta(\alpha)\phi_{n}(\alpha)=\alpha+(\xi(\alpha)-\alpha))\phi_{n}(\alpha).$ 
We may take $\alpha_{\pm n}=\pm n,$ then.  We also need to discuss the approximate vortex sheet strength, and we let
$\tilde{\gamma}^{n}$ be given by $\tilde{\gamma}^{n}(\alpha)=\tilde{\gamma}(\alpha)\phi_{n}(\alpha).$

In general, for $n\in\mathbb{N}$ we let 
\begin{equation}\nonumber
\alpha_{n}=\min\{\alpha\in\mathbb{R}:\mathrm{Re}(z(\alpha))\geq n\},
\qquad
\alpha_{-n}=\max\{\alpha\in\mathbb{R}:\mathrm{Re}(z(\alpha))\leq -n\}.
\end{equation}
Then for $\alpha\in[\alpha_{-n},\alpha_{n}],$ we define $\xi^{n}(\alpha)=\xi(\alpha).$  The approximate curve can then be further specified by
smoothly extending this curve to the remaining values of $\alpha\in\mathbb{R},$ such that the curve decays at horizontal infinity.

Beyond needing to define such an approximation, we also need convergence of a certain integral.  
We reiterate our assumptions on the approximation sequence and state the assumption on the convergence of the integral
in the following assumption.

\begin{assumption}\label{assumptionC}
We say that the curve $\mathcal{C}$ and vortex sheet strength $\gamma$ satisfy Assumption \ref{assumptionC} if there exist
sequences $\alpha_{n},$ $\alpha_{-n},$ $\xi^{n},$ and $\tilde{\gamma}^{n}$ such that 
\begin{itemize}
\item $\alpha_{n}\rightarrow\infty$ and $\alpha_{-n}\rightarrow-\infty$ as $n\rightarrow\infty,$
\item for all $n\in\mathbb{N},$
for all $\alpha\in[\alpha_{-n},\alpha_{n}],$ $\xi^{n}(\alpha)=\xi(\alpha)$ and $\tilde{\gamma}^{n}(\alpha)=\tilde{\gamma}(\alpha),$
\item  for all $n\in\mathbb{N},$ $\xi^{n}(\alpha)-\alpha\in L^{2}(\mathbb{R}),$ $\tilde{\gamma}^{n}\in L^{2}(\mathbb{R}),$ 
\item for all $n\in\mathbb{N},$ $\xi^{n}$ is twice continuously differentiable,
\item for all $n\in\mathbb{N},$ $\xi^{n}$ satisfies the chord-arc assumption \eqref{chordArc},
\item for any compact set $K_{0}\subset\mathbb{C},$ there exists $N\in\mathbb{N}$ and $A_{\pm}\in\mathbb{R}$ such that for all $n\geq N,$
for all $\alpha\notin[A_{-},A_{+}],$ $\xi^{n}(\alpha)\notin K_{0},$ and  
\item letting $\Gamma^{n}(\alpha)=\int_{0}^{\alpha}\tilde{\gamma}^{n}(\alpha')\ \mathrm{d}\alpha',$
and  $\Omega_{n}=(-\infty,\alpha_{-n}]\cup[\alpha_{n},\infty),$ for any compact set $K_{0},$ the sequence of integrals
\begin{equation}\label{integralToBeBounded}
\int_{\Omega_{n}}\frac{|\xi^{n}_{\alpha}(\alpha')||\Gamma^{n}(\alpha')|}{|z-\xi^{n}(\alpha')|^{2}}\ \mathrm{d}\alpha'
\end{equation}
converges to zero as $n$ goes to infinity, uniformly with respect to $z\in K_{0}.$
\end{itemize}
\end{assumption}

Note that this assumption appears to have many facets, but the truth is simpler than the assumption may make it appear.  For each of our examples 
listed in Section \ref{examplesSection}, we will be verifying that Assumption \ref{assumptionC} holds in Section \ref{examplesAgain} below.  All of the items in the
assumption are consequences of cutting off non-decaying functions by multiplying by a function with compact support.  We leave the situation as general
as it is given in Assumption \ref{assumptionC} so as to not limit ourselves unnecessarily.

\subsection{The main result}\label{mainTheoremSection} 
We now show that under our assumptions, the new expression for the velocity field is indeed the velocity field
associated to a vortex sheet in an incompressible flow.

\begin{theorem}\label{mainTheorem} 
If Assumption \ref{assumptionA} and Assumption \ref{assumptionC} are satisfied, and if either Assumption \ref{assumptionB} is satisfied or 
$c[\gamma]=0,$ and if the parameterization $\xi$ satisfies the following conditions:
\begin{itemize}
\item $\xi$ satisfies the chord-arc condition \eqref{chordArc},
\item  there exists $M>0$ such that for all $\alpha\in\mathbb{R},$ 
$\frac{1}{M}\leq |\xi_{\alpha}(\alpha)|\leq M,$ and
\item the function $\xi(\alpha)-\alpha$ may be decomposed as $f_{1}+f_{2}$ with $f_{1}\in L^{2}(\mathbb{R})$ and $f_{2}\in L^{\infty}(\mathbb{R}),$
\end{itemize}
then the velocity field as defined in \eqref{newVelocityForm} is incompressible and has vorticity equal to $\gamma\delta_{C}.$
Furthermore, this velocity is independent of the choice of cutoff function.
\end{theorem}

\begin{proof}
Given the family of curves $\xi^{n}$ we define an approximate velocity $(u^{n},v^{n}):$
\begin{multline}\nonumber
u^{n}-iv^{n}=\frac{1}{2\pi i}\mathrm{PV}_{\infty}\int_{-\infty}^{\infty}\frac{c[\gamma]}{z-\xi(\alpha')}\ \mathrm{d}\alpha'
\\
+\frac{1}{2\pi i}\int_{-\infty}^{\infty}\frac{a(|z-\xi^{n}(\alpha')|^{2})\tilde{\gamma}^{n}(\alpha')}{z-\xi^{n}(\alpha')}\ \mathrm{d}\alpha'
\\
-\frac{1}{2\pi i}\int_{-\infty}^{\infty}\frac{(a'(|z-\xi^{n}(\alpha')|^{2}))\left(\mathrm{Re}(\xi^{n}_{\alpha}(\alpha')(\overline{z-\xi^{n}(\alpha')}))\right)
\Gamma^{n}(\alpha')}{z-\xi^{n}(\alpha')}\ \mathrm{d}\alpha'.
\\
-\frac{1}{2\pi i}\int_{-\infty}^{\infty}\frac{(1-a(|z-\xi^{n}(\alpha')|^{2})\xi^{n}_{\alpha}(\alpha')\Gamma^{n}(\alpha')}{(z-\xi^{n}(\alpha'))^{2}}\ \mathrm{d}\alpha'.
\end{multline}
We again decompose this, as 
\begin{equation}\nonumber
u^{n}-iv^{n}=I_{0}+I_{1}^{n}+I_{2}^{n}+I_{3}^{n}.
\end{equation}
Note that $I_{0}$ does not have a superscript here because it is unmodified from the corresponding expression in \eqref{newVelocityForm}.

We will study the convergence of $(u^{n},v^{n})$ to $(u,v),$ but first we want to consider $(u^{n},v^{n})$ in its own right.  
Considering for the moment only $I_{2}^{n}+I_{3}^{n},$ we can write
\begin{equation}\label{reintegrateByParts}
I_{2}^{n}+I_{3}^{n}
=\frac{1}{2\pi i}\int_{-\infty}^{\infty}\partial_{\alpha'}\left(\frac{(1-a(|z-\xi^{n}(\alpha')|^{2}))}{z-\xi^{n}(\alpha')}\right)\tilde{\Gamma}^{n}(\alpha')\ \mathrm{d}\alpha'.
\end{equation}
Since $\tilde{\gamma}^{n}\in L^{2}(\mathbb{R}),$ by Lemma \ref{assumptionASatisfiedLemma} and its proof, 
we have that $\tilde{\Gamma}^{n}(\alpha')$ grows at infinity only like $\sqrt{\alpha'}$ at worst.
Thus, we may integrate \eqref{reintegrateByParts} by parts with zero boundary contributions.  When combining this with the remaining parts of the velocity, we find
\begin{equation}\label{noCutoff}
u_{n}-iv_{n}=I_{0}+\frac{1}{2\pi i}\mathrm{PV}\int_{-\infty}^{\infty}\frac{\tilde{\gamma}^{n}(\alpha')}{z-\xi^{n}(\alpha')}\ \mathrm{d}\alpha'.
\end{equation}
This is the velocity field for a decaying vortex sheet with vortex sheet strength equal to $\tilde{\gamma}^{n}+c[\gamma],$ with the position of the sheet
given by the curve $\mathcal{C}^{n}$ whose position in the complex plane is given by $\xi^{n}.$  Thus $(u_{n},v_{n})$ is an incompressible
velocity field, with vorticity equal to $c[\gamma]\delta_{\mathcal{C}}+\tilde{\gamma}^{n}\delta_{\mathcal{C}_{n}}.$

Let $K\subset\mathbb{R}^{2}$ be a compact set.  We now wish to show that $u^{n}-iv^{n}$ converges uniformly to $u-iv$ on $K.$
We will accomplish this by showing that $I_{1}^{n}$ converges uniformly to $I_{1}$ on $K,$ $I_{2}^{n}$ converges uniformly to $I_{2}$ on $K,$
and $I_{3}^{n}$ converges uniformly to $I_{3}$ on $K.$

We define a second compact set, $K_{3},$ to be the set of all points $p\in\mathbb{R}^{2}$ such that $d(p,K)\leq3.$  Then there exists
$A_{\pm}\in\mathbb{R}$ such that for $\alpha\geq A_{+}$ or for $\alpha<A_{-},$ $\xi(\alpha)\notin K_{3}.$ 
Furthermore, by our assumptions on the approximations $\xi^{n},$ 
there exists $N\in\mathbb{N}$ such that for all $n\geq N,$ for all $\alpha\in[A_{-},A_{+}],$ $\xi^{n}(\alpha)=\xi(\alpha)$
and $\tilde{\gamma}^{n}(\alpha)=\tilde{\gamma}(\alpha),$
and for all $\alpha\notin[A_{-},A_{+}],$ $\xi^{n}(\alpha)\notin K_{3}.$
Let $z\in K$ be given.  Then $I_{1}$ and $I_{1}^{n}$ can be rewritten as
\begin{equation}\nonumber
I_{1}=\frac{1}{2\pi i}\int_{A_{-}}^{A_{+}}\frac{a(|z-\xi(\alpha')|^{2})\tilde{\gamma}(\alpha')}{z-\xi(\alpha')}\ \mathrm{d}\alpha',
\end{equation}
\begin{equation}\nonumber
I_{1}^{n}=\frac{1}{2\pi i}\int_{A_{-}}^{A_{+}}\frac{a(|z-\xi^{n}(\alpha')|^{2})\tilde{\gamma}^{n}(\alpha')}{z-\xi^{n}(\alpha')}\ \mathrm{d}\alpha',
\end{equation} 
for $n\geq N.$
But then we see immediately that for all $n\geq N,$ we have $I_{1}=I_{1}^{n}.$  So, indeed, $I_{1}^{n}$ converges
uniformly on $K$ to $I_{1},$    The same situation holds for $I_{2}^{n}$ converging uniformly on $K$ to $I_{2}^{n}.$

It only remains to show that $I_{3}^{n}$ converges uniformly on $K$ to $I_{3}.$  Note that 
we only need to consider the domain of integration to be $(-\infty,\alpha_{-n}]\cup[\alpha_{n},\infty),$ as the integrands
for $I_{3}$ and $I_{3}^{n}$ agree on the complement of this domain.  We will only consider $[\alpha_{n},\infty)$ in detail,
as the other case is entirely similar.  We define
\begin{equation}\nonumber
J^{n}=-\frac{1}{2\pi i}\int_{\alpha_{n}}^{\infty}
\frac{(1-a(|z-\xi(\alpha')|^{2})\xi_{\alpha}(\alpha')\Gamma(\alpha')}{(z-\xi(\alpha'))^{2}}\ \mathrm{d}\alpha',
\end{equation}
\begin{equation}\nonumber
I_{3}^{n,+}=-\frac{1}{2\pi i}\int_{\alpha_{n}}^{\infty}
\frac{(1-a(|z-\xi^{n}(\alpha')|^{2})\xi^{n}_{\alpha}(\alpha')\Gamma^{n}(\alpha')}{(z-\xi^{n}(\alpha'))^{2}}\ \mathrm{d}\alpha'.
\end{equation}
We need to show that $J^{n}-I_{3}^{n,+}$ can be made small by taking $n$ large, uniformly with respect to $z\in K.$  
We will accomplish this by showing that the two terms may be made individually small by taking $n$ large, uniformly with respect to $z\in K.$

For $n$ sufficiently large, for $\alpha'>\alpha_{n},$ it will be the case that $a(|z-\xi(\alpha')|^{2})=0.$  Thus we may bound $|J_{n}|$ as
\begin{equation}\nonumber
|J_{n}|\leq \int_{\alpha_{n}}^{\infty}\frac{|\xi_{\alpha}(\alpha')||\Gamma(\alpha')|}{|z-\xi(\alpha')|^{2}}\ \mathrm{d}\alpha'.
\end{equation}
Using the assumption that $|\xi_{\alpha}|$ is bounded and that $\gamma$ satisfies Assumption \ref{assumptionA}, we may further bound this as
\begin{equation}\nonumber
|J_{n}|\leq c \int_{\alpha_{n}}^{\infty} \frac{|\alpha'|^{\beta}}{|z-\xi(\alpha')|^{2}}\ \mathrm{d}\alpha',
\end{equation}
where the constant $c$ and the power $\beta\in[0,1)$ are taken to be independent of $n.$  
We let $\alpha_{*}$ be such that $\xi(\alpha_{*})$ is the nearest point on $\mathcal{C}$ to $z.$  Then we have
\begin{equation}\nonumber
|J_{n}|\leq c\int_{\alpha_{n}}^{\infty}\frac{|\alpha'|^{\beta}}{|\xi(\alpha_{*})-\xi(\alpha')|^{2}}\ \mathrm{d}\alpha'.
\end{equation}
Using the chord-arc condition \eqref{chordArc}, we may then bound this as
\begin{equation}\nonumber
|J_{n}|\leq c\int_{\alpha_{n}}^{\infty}\frac{|\alpha'|^{\beta}}{(\alpha'-\alpha_{*})^{2}}\ \mathrm{d}\alpha'.
\end{equation}
Since $z$ is chosen from a compact set, the set of all such $\alpha_{*}$ may also be taken from a compact set; we use the name $\alpha_{**}$ for
the largest possible such $\alpha_{*}.$  Then we have
 \begin{equation}\label{JnNoZ}
|J_{n}|\leq c\int_{\alpha_{n}}^{\infty}\frac{|\alpha'|^{\beta}}{(\alpha'-\alpha_{**})^{2}}\ \mathrm{d}\alpha'.
\end{equation}
Since $\alpha_{n}\rightarrow\infty$ as $n\rightarrow\infty,$ we conclude that $J_{n}\rightarrow0$ as $n\rightarrow\infty.$  Furthermore, since
\eqref{JnNoZ} no longer has any $z$ dependence on the right-hand side, this convergence is seen to be uniform with respect to the choice of
$z\in K.$

That the same conclusion holds for $I_{3}^{n,+}$ follows analogously, relying upon Assumption \ref{assumptionC}.  This concludes the proof that
$(u_{n},v_{n})$ converges to $(u,v)$ on compact sets.  We next must demonstrate the remaining claims, namely that $(u,v)$ has the desired divergence
and curl, and is independent of the choice of cutoff function, $a.$

With the velocity $(u_{n},v_{n})$ converging to $(u,v)$ uniformly on compact sets, we conclude that $(u_{n},v_{n})$ converges
to $(u,v)$ in the sense of distributions.  That is, for any pair of test functions $(\phi_{1},\phi_{2}),$ which is a pair of $C^{\infty}$ 
functions with compact support, we have
\begin{equation}\nonumber
((u_{n},v_{n}),(\phi_{1},\phi_{2}))
=\int_{X}u_{n}\phi_{1}+v_{n}\phi_{2} \rightarrow \int_{X}u\phi_{1}+v\phi_{2} = ((u,v),(\phi_{1},\phi_{2})).
\end{equation}
Here, $X$ is a compact set containing the support of $\phi_{1}$ and $\phi_{2}.$  
Furthermore, for any test function $\phi,$ this implies 
\begin{equation}\nonumber
(u_{n},\phi_{x})\rightarrow (u,\phi_{x}),\qquad (u_{n},\phi_{y})\rightarrow (u,\phi_{y}),
\end{equation}
\begin{equation}\nonumber
(v_{n},\phi_{x})\rightarrow (v,\phi_{x}),\qquad (v_{n},\phi_{y})\rightarrow (v,\phi_{y}).
\end{equation}
This implies that $\mathrm{curl}(u_{n},v_{n})$ converges to $\mathrm{curl}(u,v)$ in the sense of distributions, and 
$\mathrm{div}(u_{n},v_{n})$ converges to $\mathrm{div}(u,v)$ in the sense of distributions.
But $\mathrm{div}(u_{n},v_{n})=0.$  We conclude that $(u,v)$ is divergence-free, in the
sense of distributions.
Note that it is not possible that the velocity field would 
be divergence-free in a classical sense, as the velocity is not continuous
and therefore not differentiable in the classical sense.  As we have said, the vorticity associated to $(u_{n},v_{n})$ is 
$c[\gamma]\delta_{\mathcal{C}}+\tilde{\gamma}^{n}\delta_{C_{n}},$ by our construction.  This means that $(u_{n},v_{n})$ is irrotational in the
interior of the two fluid regions.  Since the boundary curves $\xi_{n}$ converge uniformly (eventually identically, in fact) to the curves 
$\xi$ on compact sets, and similarly for $\tilde{\gamma}^{n},$
we may conclude that the vorticity associated to the velocity field $(u,v)$ is $\gamma\delta_{C},$ 
and again this implies the the fluid is irrotational in the interior of either fluid region.

Finally, we discuss independence of the choice of cutoff function.  As can be seen in \eqref{noCutoff}, the approximate
velocities $(u_{n},v_{n})$ do not depend on the choice of the cutoff function, $a.$  Since these converge to the velocity
$(u,v),$ it is not possible for the velocity $(u,v)$ to depend on this choice.
\end{proof}

\subsection{Jump conditions}\label{jumpConditionsSection}

We may express the limiting values of the velocity $(u,v)$ as the point $\mathbf{x}$ approaches $\mathcal{C}$ in terms of $\mathbf{W}.$
We let $\mathbf{\hat{t}}$ and $\mathbf{\hat{n}}$ be the usual frame of unit normal and tangent vectors to the curve,
\begin{equation}\nonumber
\mathbf{\hat{t}}=\frac{(x_{\alpha},y_{\alpha})}{s_{\alpha}},
\qquad
\mathbf{\hat{n}}=\frac{(-y_{\alpha},x_{\alpha})}{s_{\alpha}},
\end{equation}
where the arclength element $s_{\alpha}$ is given by $s_{\alpha}=\sqrt{x_{\alpha}^{2}+y_{\alpha}^{2}}.$
We wish to take the limit of \eqref{newVelocityForm} as $z$ approaches $\mathcal{C},$ and so we take the limit of each of 
$I_{0},$ $I_{1},$ $I_{2},$ and $I_{3}.$  Since the integrals $I_{2}$ and $I_{3}$ are non-singular, the limit of these is just
$W_{2}$ and $W_{3},$ when $z\rightarrow\xi(\alpha).$  For $I_{0}$ and $I_{1}$ we are able to 
apply the Plemelj formulas (see \cite{plemelj}, and see also \cite{bakerMeironOrszag}) 
when taking the limit as $z$ approaches $\mathcal{C}.$  This is just as in the usual decaying case of vortex sheets, and we find
\begin{equation}\label{limitingBehavior}
(u,v)\rightarrow W \pm \frac{\gamma}{2s_{\alpha}}\mathbf{\hat{t}},
\end{equation}
where the plus or minus depends on whether the interface is approached from above or below.
We see from \eqref{limitingBehavior} that there is no jump in normal velocity across the vortex sheet, while the jump
in tangential velocity is given by the vortex sheet strength, $\gamma.$

\subsection{Limit at vertical infinity}\label{verticalInfinitySection}

It is of interest for a vortex sheet to compute the limiting velocity at vertical infinity, i.e. to take the limit of \eqref{newVelocityForm} as $y\rightarrow\pm\infty.$
It is immediate to see that as $y\rightarrow\pm\infty,$ $I_{1}$ and $I_{2}$ become identically zero because of the presence of the cutoff function.
Thus we only must work to take the limit of $I_{0}$ and $I_{3}.$  In fact, only $I_{0}$ makes a nonzero contribution to the limit.
We may write $I_{0}$ as
\begin{equation}\nonumber
I_{0}=\frac{c[\gamma]}{2\pi i}\mathrm{PV}_{\infty}\int_{-\infty}^{\infty}\frac{1}{z-\alpha'}\ \mathrm{d}\alpha'
+\frac{c[\gamma]}{2\pi i}\mathrm{PV}_{\infty}\int_{-\infty}^{\infty}\frac{\xi(\alpha')-\alpha'}{(z-\alpha')(z-\xi(\alpha'))}\ \mathrm{d}\alpha'.
\end{equation}
The first term on the right-hand side has limit $\mp\frac{c[\gamma]}{2}$ as $y\rightarrow\pm\infty;$ notice this integral is now in the decaying case, and the calculation
follows from being able to exactly evaluate the integral.  The second integral on the right-hand side has limit zero at vertical infinity, which can be seen because the
denominator is now quadratic with respect to the $z$ variable.  Similarly, the limit of $I_{3}$ as $y\rightarrow\pm\infty$ is zero, because the denominator is quadratic
with respect to $z.$   We conclude that
\begin{equation}\nonumber
\lim_{y\rightarrow\pm\infty}(u,v)(x,y)=\mp\frac{c[\gamma]}{2},
\end{equation}
which is the same limit as in the periodic and decaying cases.

\section{Revisiting the examples}\label{examplesAgain}
 In this section we consider again the four illustrative examples of curves $\mathcal{C}$ and vortex sheet strengths $\gamma$ which we introduced
 in Section \ref{examplesSection} above.  We will see that for all of these examples, all the needed assumptions are satisfied.
 
In each of the four examples of Section \ref{examplesSection}, we let
the approximate curve $\mathcal{C}_{n}$ be given by its parameterization $\xi^{n}(\alpha)=\alpha+(\xi(\alpha)-\alpha))\phi_{n}(\alpha).$ 
To see that Theorem \ref{mainTheorem} applies in each of the four examples, we must verify Assumption \ref{assumptionA} and 
Assumption \ref{assumptionC}, and either we must have that $c[\gamma]=0$ or we must verify Assumption \ref{assumptionB}.  Furthermore, the bullet-point list
of conditions on the parameterization $\xi$ which appears in the theorem must be verified.  Most of these points to verify are obvious and will not be commented 
on explicitly.  Below we give some further details.

\begin{description}
\item[(a) A mixed case.]\ We remark that the vortex sheet strength $\gamma(\alpha)=\left(\frac{1}{1+\alpha^{2}}\right)^{2/5}$ satisfies Assumption
\ref{assumptionA} with $c[\gamma]=0$ and $\beta=1/5.$  Since $c[\gamma]=0,$ we do not need to consider Assumption \ref{assumptionB}.
We only need to verify the condition on the integral in \eqref{integralToBeBounded}.  If we call the integral $\mathcal{I}_{n}(z),$ we have the estimate, for some
$c$ independent of $n$ and $z,$
\begin{equation}\nonumber
|\mathcal{I}_{n}(z)|\leq c\int_{\Omega_{n}} \frac{\alpha'^{1/5}}{(x-\xi^{n}_{1}(\alpha'))^{2}+(y-\xi^{n}_{2}(\alpha'))^{2}}\ \mathrm{d}\alpha'.
\end{equation}
We have written here $z=x+iy.$  We let $x_{*+}$ be the largest $x$-value in the compact set $K_{0},$ and we let $x_{*-}$ be the smallest $x$-value in $K_{0}.$
We can use our specific parameterization,
namely $\xi^{n}_{1}(\alpha)=\alpha,$  and we neglect the term $(y-\xi^{n}_{2}(\alpha'))^{2}$ in the denominator.
We identify the values $\alpha_{n}=n$ and $\alpha_{-n}=-n.$
 These considerations lead to the bound
\begin{equation}\nonumber
|\mathcal{I}_{n}(z)|\leq c\int_{-\infty}^{-n}\frac{\alpha'^{1/5}}{(x_{*-}-\alpha')^{2}}\ \mathrm{d}\alpha'
+c\int_{n}^{\infty}\frac{\alpha'^{1/5}}{(x_{*+}-\alpha')^{2}}\ \mathrm{d}\alpha'.
\end{equation}
Note that for sufficiently large $n,$ it is not possible for the denominators of the integrands to equal zero.
These integrals do converge to zero as $n$ goes to infinity, and this is clearly uniform with respect to $z.$

\

\item[(b) A quasiperiodic case.]\  The vortex sheet strength $\gamma(\alpha)=\cos\left(\left(1+4\pi-\sqrt{2}\right)\alpha\right)$ satisfies Assumption
\ref{assumptionA} with $c[\gamma]=0$ and $\beta=0.$  Since $c[\gamma]=0,$ we do not need to consider Assumption \ref{assumptionB}.
In this case we need to pay special attention to the chord-arc condition \eqref{chordArc}, for $\alpha$ and $\alpha'$ close to each other.  It can
be seen that the chord-arc quantity (i.e. the divided difference $q_{1}[\xi]$) is equal to $1+O(\mu),$ for $\alpha-\alpha'$ small.  Thus, for $\mu$ sufficiently
small, the chord-arc condition is satisfied and the curve is non-self-intersecting.  The other properties are similar to those in example (a).

\

\item[(c) A uniformly local case.]\  The vortex sheet strength $\gamma(\alpha)=\frac{\alpha^{2}}{1+\alpha^{2}}$ satsifes Assumption \ref{assumptionA}
with $c[\gamma]=1$ and $\beta=0.$  We next consider Assumption \ref{assumptionB}, and we argue as in the proof of Lemma \ref{assumptionBLemma}.
We follow the proof of the lemma until \eqref{toVerifyAssumptionB1}; then, it is clear that the integral on the right-hand side of \eqref{toVerifyAssumptionB1}
converges with $\xi(\alpha)=\alpha+i\sin\left(\left(\frac{\alpha^{4}}{1+\alpha^{2}}\right)^{1/4}\right),$ since the numerator of the second factor is bounded and the 
denominator is of the order $\alpha'^{2}.$  To bound this, we must also use the fact that the parameterization $\xi$ satisfies the chord-arc condition \eqref{chordArc},
which it does.  We then continue to follow the proof of the lemma until \eqref{toVerifyAssumptionB2}, and we again note that the integral on the right-hand side
converges since the numerator is bounded and the denominator is of the order $\alpha'^{2}.$  The other properties are again similar to those in example (a).

\

\item[(d) A bore.]\ The vortex sheet strength $\gamma(\alpha)=1+\sin(\alpha)$ satisfies Assumption \ref{assumptionA} with $c[\gamma]=1$ and
$\beta=0.$  That Assumption \ref{assumptionB} holds is verified in the same manner as in example (c), and the remaining properties are similar to those in
example (a).  

While in the case of the bore we took the location of the vortex sheet to have different limits at positive and negative infinity, it is worth noting that the same
cannot be done for the vortex sheet strength.  If, say, $\gamma$ approached $-1$ at negative infinity and $1$ at positive infinity, then we would not be able
to satisfy Assumption \ref{assumptionA}.

\end{description}

\section{Discussion}\label{discussionSection}
In this work, we have developed a more general expression for the velocity field associated to a vortex sheet than the one
found by directly using the vortex sheet ansatz for vorticity in the Biot-Savart law.  Taking the limit of this velocity field as the 
vortex sheet is approached yields a new, more general expression for the Birkhoff-Rott integral.  The Birkhoff-Rott integral has the advantages
that it is useful for analytical and numerical studies of the motion of vortex sheets (including water waves) in situations in which the
sheet is overturned, i.e. has multi-valued height, and also that can be used in this way for both two-dimensional and three-dimensional
fluid flows.  While we have only treated the case of two-dimensional fluids here, we expect that this work can be carried over to the three-dimensional
case.  A development of the usual Birkhoff-Rott integral for three-dimensional fluids may be found in \cite{caflischLi}.

In future work, we expect to categorize the set of interface locations for which Assumption \ref{assumptionC} on approximability by decaying
sheets holds.  We also expect to 
use the new expression for the Birkhoff-Rott integral to prove well-posedness theory for 
vortex sheets and water waves in the non-decaying, non-periodic context.  The author and Masmoudi have previously made such studies
in two-dimensional fluids in the periodic setting (\cite{ambroseThesis}, \cite{ambroseMasmoudi1}), and for three-dimensional fluids
in the decaying setting (\cite{ambroseMasmoudi2}, \cite{ambroseMasmoudi3}).  In these works a key ingredient has been approximation
of the Birkhoff-Rott integral by the Hilbert transform or by Riesz transforms; this approach has also been employed in similar works
such as \cite{BHLGrowthRates} or \cite{cordoba}.  If a function $f$ satisfies Assumption \ref{assumptionA},
we may define its Hilbert transform as
\begin{multline}\label{newHilbert}
H(f)(\alpha)=\frac{1}{\pi}\mathrm{PV}\int_{-\infty}^{\infty}\frac{a(\alpha-\alpha')f(\alpha')}{\alpha-\alpha'}\ \mathrm{d}\alpha'
-\frac{1}{\pi}\int_{-\infty}^{\infty}\frac{a'(\alpha-\alpha')F(\alpha')}{\alpha-\alpha'}\ \mathrm{d}\alpha'
\\
-\frac{1}{\pi}\int_{-\infty}^{\infty}\frac{(1-a(\alpha-\alpha'))F(\alpha')}{(\alpha-\alpha')^{2}}\ \mathrm{d}\alpha'.
\end{multline}
Here, $F$ is the antiderivative of $f-c[f]$ discussed in Assumption \ref{assumptionA}.
The formula \eqref{newHilbert} is developed analogously to the development of the new Birkhoff-Rott integral.  Just as in that case, this is
a single new formula for the Hilbert transform which unifies the decaying and periodic settings while extending its reach to 
additional functions.  Of course, the Hilbert transform has additional properties such as connections with analytic functions, 
and we will leave it to future work to develop such theory for this new version of the Hilbert transform.  This future work will 
also prioritize demonstrating the necessary estimates for approximation of the new Birkhoff-Rott integral by the new 
Hilbert transform.

\appendix

\section{Proofs of lemmas}\label{proofsSection}

In this appendix. we give the proofs of two lemmas.  First we have the proof of Lemma \ref{assumptionASatisfiedLemma}, about Assumption \ref{assumptionA}.

\begin{proof}[{\bf Proof of Lemma \ref{assumptionASatisfiedLemma}}] Let $\gamma\in L^{2}(\mathbb{R}).$  We define $\Gamma$ by
\begin{equation}\nonumber
\Gamma(\alpha)=\int_{0}^{\alpha}\gamma(\alpha')\ \mathrm{d}\alpha'.
\end{equation}
Then by H\"{o}lder's inequality (\cite{kreyszig}), we have for any $\alpha,$
\begin{equation}\nonumber
|\Gamma(\alpha)|\leq \left(\int_{0}^{\alpha}1^{2}\ \mathrm{d}\alpha'\right)^{1/2}\left(\int_{0}^{\alpha}\gamma^{2}(\alpha')\ \mathrm{d}\alpha'\right)^{1/2}
\leq \sqrt{\alpha}\|\gamma\|_{L^{2}(\mathbb{R})}.
\end{equation}
We conclude that Assumption \ref{assumptionA} is satisfied with $C=\|\gamma\|_{L^{2}(\mathbb{R})}$ and $\beta=1/2,$
and with constant $c[\gamma]=0.$

Now assume instead that $\gamma\in L^{2}(\mathbb{T}).$  Then we may express $\gamma$ by its Fourier series,
\begin{equation}\nonumber
\gamma(\alpha)=\sum_{k\in\mathbb{Z}}\hat{\gamma}(k)e^{ik\alpha}.
\end{equation}
(Here we have again set the periodicity length to be $2\pi,$ but this is not necessary.)  Then we may identify the constant
$c[\gamma]=\hat{\gamma}(0).$  The antiderivative $\Gamma$ may be taken to be 
\begin{equation}\nonumber
\Gamma(\alpha)=\sum_{k\in\mathbb{Z}\setminus\{0\}}\frac{\hat{\gamma}(k)}{ik}e^{ik\alpha}.
\end{equation}
Since $\gamma$ is in $L^{2}(\mathbb{T}),$ by Plancherel's theorem we may conclude $\Gamma\in L^{2}(\mathbb{T})$ as well:
\begin{equation}
\|\Gamma\|_{L^{2}(\mathbb{T})}=\sum_{k\in\mathbb{Z}\setminus\{0\}}\frac{|\hat{\gamma}(k)|^{2}}{k^{2}}
\leq\sum_{k\in\mathbb{Z}}|\hat{\gamma}(k)|^{2}=\|\gamma\|_{L^{2}(\mathbb{T})}.
\end{equation}
By Sobolev embedding (\cite{adams}), this implies $\Gamma\in L^{\infty}.$
We conclude that Assumption \ref{assumptionA} is satisfied, with $\beta=0.$
\end{proof}

Finally, we have the proof of Lemma \ref{assumptionBLemma}, about Assumption \ref{assumptionB}.

\begin{proof}[{\bf Proof of Lemma \ref{assumptionBLemma}}] 
We begin by considering item (i) of Assumption \ref{assumptionB}.
We decompose the integral as $N+F,$ (standing for the 
``near'' part and the ``far'' part), with
\begin{equation}\nonumber
N=\mathrm{PV}\int_{\alpha-1}^{\alpha+1}\frac{1}{\xi(\alpha)-\xi(\alpha')}\ \mathrm{d}\alpha',
\end{equation}
\begin{equation}\nonumber
F=\mathrm{PV}_{\infty}\left(\int_{-\infty}^{\alpha-1}\frac{1}{\xi(\alpha)-\xi(\alpha')}\ \mathrm{d}\alpha'
+\int_{\alpha+1}^{\infty}\frac{1}{\xi(\alpha)-\xi(\alpha')}\ \mathrm{d}\alpha'\right).
\end{equation}

We first deal with $N.$  Adding and subtracting, we can write this as
\begin{equation}\label{lastN}
N=\mathrm{PV}\int_{\alpha-1}^{\alpha+1}\frac{1}{\xi_{\alpha}(\alpha')(\alpha-\alpha')}\ \mathrm{d}\alpha'
+\int_{\alpha-1}^{\alpha+1}\left[\frac{1}{\xi(\alpha)-\xi(\alpha')}-\frac{1}{\xi_{\alpha}(\alpha')(\alpha-\alpha')}\right]\ \mathrm{d}\alpha'.
\end{equation}
The first integral on the right-hand side converges, as this is a localized Hilbert transform.  For the second integral on the right-hand side,
it is the integral of a ratio of divided differences.  That is, we let
\begin{equation}\nonumber
q_{1}[\xi](\alpha,\alpha')=\frac{\xi(\alpha)-\xi(\alpha')}{\alpha-\alpha},
\qquad
q_{2}[\xi](\alpha,\alpha')=\frac{\xi(\alpha)-\xi(\alpha')-\xi_{\alpha}(\alpha')(\alpha-\alpha')}{(\alpha-\alpha')^{2}};
\end{equation}
these are known as the first and second divided differences for $z.$  Then the second integral on the right-hand side of \eqref{lastN}
is the integral of $\frac{q_{2}[\xi](\alpha,\alpha')}{\xi_{\alpha}(\alpha')q_{1}[\xi](\alpha,\alpha')}.$  By standard results on divided differences
(as in e.g. \cite{BHLGrowthRates}), if $\xi$ is twice differentiable then $q_{2}$ is locally integrable.
Using together our assumption that $\xi$ is twice differentiable, the chord-arc condition 
\eqref{chordArc} and the stated assumption on the parameterization that $\xi_{\alpha}$ is bounded away from zero, we may conclude that
the integral $N$ converges.

We next consider $F.$  In each of the integrands in $F,$ we subtract $\frac{1}{\alpha-\alpha'}:$
\begin{multline}\nonumber
F=\mathrm{PV}_{\infty}\Bigg(\int_{-\infty}^{\alpha-1}\frac{1}{\alpha-\alpha'}\ \mathrm{d}\alpha'
+\int_{-\infty}^{\alpha-1}\frac{1}{\xi(\alpha)-\xi(\alpha')}-\frac{1}{\alpha-\alpha'}\ \mathrm{d}\alpha'
\\
+\int_{\alpha+1}^{\infty}\frac{1}{\alpha-\alpha'}\ \mathrm{d}\alpha'
+\int_{\alpha+1}^{\infty}\frac{1}{\xi(\alpha)-\xi(\alpha')}-\frac{1}{\alpha-\alpha'}\ \mathrm{d}\alpha'
\Bigg).
\end{multline}
We can clearly compute part of this principal value:
\begin{equation}\nonumber
\mathrm{PV}_{\infty}\left(\int_{-\infty}^{\alpha-1}\frac{1}{\alpha-\alpha'}\ \mathrm{d}\alpha'
+\int_{\alpha+1}^{\infty}\frac{1}{\alpha-\alpha'}\ \mathrm{d}\alpha'\right)=0.
\end{equation}
This leaves us with the following formula for $F:$
\begin{equation}\nonumber
F=\mathrm{PV}_{\infty}\Bigg(\int_{-\infty}^{\alpha-1}\frac{1}{\xi(\alpha)-\xi(\alpha')}-\frac{1}{\alpha-\alpha'}\ \mathrm{d}\alpha'
+\int_{\alpha+1}^{\infty}\frac{1}{\xi(\alpha)-\xi(\alpha')}-\frac{1}{\alpha-\alpha'}\ \mathrm{d}\alpha'\Bigg).
\end{equation}
This is no longer, in fact, a principal value integral at infinity, however, since these integrals converge individually.
To see this, we arrange the second integral on the right-hand side as
\begin{equation}\label{toVerifyAssumptionB1}
\int_{\alpha+1}^{\infty}\frac{1}{\xi(\alpha)-\xi(\alpha')}-\frac{1}{\alpha-\alpha'}\ \mathrm{d}\alpha'
=-\int_{\alpha+1}^{\infty}\frac{\alpha-\alpha'}{\xi(\alpha)-\xi(\alpha')}
\frac{(\xi(\alpha)-\alpha)-(\xi(\alpha')-\alpha')}{(\alpha-\alpha')^{2}}
\ \mathrm{d}\alpha'.
\end{equation}
We then decompose this into two integrals, as
\begin{multline}\nonumber
\int_{\alpha+1}^{\infty}\frac{1}{\xi(\alpha)-\xi(\alpha')}-\frac{1}{\alpha-\alpha'}\ \mathrm{d}\alpha'
=-(\xi(\alpha)-\alpha)\int_{\alpha+1}^{\infty}\frac{\alpha-\alpha'}{\xi(\alpha)-\xi(\alpha')}\frac{1}{(\alpha-\alpha')^{2}}\ \mathrm{d}\alpha'
\\
+\int_{\alpha+1}^{\infty}\frac{\alpha-\alpha'}{\xi(\alpha)-\xi(\alpha')}\frac{\xi(\alpha')-\alpha'}{(\alpha-\alpha')^{2}}\ \mathrm{d}\alpha'.
\end{multline}
That the first integral on the right-hand side converges follows from the chord-arc condition \eqref{chordArc} and from the fact that 
$(\alpha-\alpha')^{-2}$ is integrable on the relevant domain.  That the second integral on the right-hand side converges follows
from these facts plus the regularity of $\xi(\alpha')-\alpha'.$  The other integral comprising $F$ may be treated in just the same
way, and the conclusion for the case $\xi(\alpha)-\alpha$ is decaying now follows.

For the case that $\xi(\alpha)-\alpha$ is periodic, the principal value at positive and negative infinity is implicit in the 
cotangent summation formula discussed above.  For the principal value at $\alpha'=\alpha,$ this is essentially the same as
in the real-line case, and we do not provide further details (although the interested reader could consult \cite{ambroseThesis}).

We now begin to prove (ii).
We write the point $z$ as $z=x+iy,$ and we write the integral as
\begin{equation}\nonumber
\mathrm{PV}_{\infty}\int_{-\infty}^{\infty}\frac{1}{z-\xi(\alpha')}\ \mathrm{d}\alpha'
=\mathrm{PV}_{\infty}\left(\int_{-\infty}^{x}\frac{1}{z-\xi(\alpha')}\ \mathrm{d}\alpha'
+\int_{x}^{\infty}\frac{1}{z-\xi(\alpha')}\ \mathrm{d}\alpha'\right).
\end{equation}
We add and subtract $\frac{1}{z-\alpha'}$ in each integral, arriving at
\begin{multline}\nonumber
\mathrm{PV}_{\infty}\int_{-\infty}^{\infty}\frac{1}{z-\xi(\alpha')}\ \mathrm{d}\alpha'
=\mathrm{PV}_{\infty}\left(\int_{-\infty}^{x}\frac{1}{z-\alpha'}\ \mathrm{d}\alpha'
+\int_{x}^{\infty}\frac{1}{z-\alpha'}\ \mathrm{d}\alpha'\right)
\\
+\int_{-\infty}^{\infty}\frac{1}{z-\xi(\alpha')}-\frac{1}{z-\alpha'}\ \mathrm{d}\alpha'.
\end{multline}
Notice that the final integral on the right-hand side is no longer a principal value (we will justify this shortly).
Also, we can compute that the first integral on the right-hand side is zero:
\begin{equation}\nonumber
\mathrm{PV}_{\infty}\left(\int_{-\infty}^{x}\frac{1}{z-\alpha'}\ \mathrm{d}\alpha'+\int_{x}^{\infty}\frac{1}{z-\alpha'}\ \mathrm{d}\alpha'\right)=0.
\end{equation}
For the remaining integral, we add the fractions, finding
\begin{equation}\label{toVerifyAssumptionB2}
\mathrm{PV}_{\infty}\int_{-\infty}^{\infty}\frac{1}{z-\xi(\alpha')}\ \mathrm{d}\alpha'
=\int_{-\infty}^{\infty}\frac{\xi(\alpha')-\alpha'}{(z-\xi(\alpha'))(z-\alpha')}\ \mathrm{d}\alpha'.
\end{equation}
In this final integral, $\frac{1}{z-\xi(\alpha')}$ is bounded away from zero since $z$ is a positive distance from $\mathcal{C}.$
If $\xi(\alpha')-\alpha'$ is in $L^{2}(\mathbb{R}),$ then since the function $\frac{1}{z-\alpha'}$ is also in $L^{2}(\mathbb{R}),$ this 
integral converges.  This completes
the proof in the case that $\xi(\alpha)-\alpha$ is in $L^{2}(\mathbb{R}).$

In the periodic case, we just need to add and subtract one more time.  We have
\begin{multline}\nonumber
\mathrm{PV}_{\infty}\int_{-\infty}^{\infty}\frac{1}{z-\xi(\alpha')}\ \mathrm{d}\alpha'
=\int_{-\infty}^{\infty}\frac{\xi(\alpha')-\alpha'}{(z-\alpha')^{2}}\ \mathrm{d}\alpha'
\\
+\int_{-\infty}^{\infty}\frac{\xi(\alpha')-\alpha'}{z-\alpha'}\left[\frac{1}{z-\xi(\alpha')}-\frac{1}{z-\alpha'}\right]\ \mathrm{d}\alpha'.
\\
=\int_{-\infty}^{\infty}\frac{\xi(\alpha')-\alpha'}{(z-\alpha')^{2}}\ \mathrm{d}\alpha'
+\int_{-\infty}^{\infty}\frac{(\xi(\alpha')-\alpha')^{2}}{(z-\xi(\alpha'))(z-\alpha')^{2}}\ \mathrm{d}\alpha'.
\end{multline}
Again, the factor
$\frac{1}{z-\xi(\alpha')}$ is bounded above since $z$ is a positive distance from $\mathcal{C}.$  Finally, the factor
$\frac{1}{(z-\alpha')^{2}}$ is integrable on $\mathbb{R}.$  Since we have assumed that $\xi(\alpha)-\alpha$ is bounded, 
this completes the proof in the periodic case.
\end{proof}

\noindent{\bf Acknowledgements}: {The author wishes to thank Mark Hoefer for suggesting a bore as an example.  The author thanks Jon Wilkening for 
helpful discussions.  The author thanks Tom Beale for comments on an early version of this manuscript.  The author thanks Michael Siegel for 
extensive comments on an early version of this manuscript.}\\

\noindent{\bf Funding}: {The author is grateful to the National Science Foundation for support through grant DMS-2307638.}\\

\noindent{\bf Declaration of interests}: {The author reports no conflict of interest.}\\

\noindent{\bf Author ORCIDs}: {D.M. Ambrose, https://orcid.org/0000-0003-4753-0319.}

\bibliographystyle{jfm}

\bibliography{nondecayingBirkhoff-Rott.bib}{}

\end{document}